\documentclass[11pt]{article}
\usepackage[T1]{fontenc}
\usepackage[utf8]{inputenc}
\usepackage[margin=1in]{geometry}

\usepackage{amsfonts,amsthm,amsmath,amssymb}
\usepackage{cite,microtype,graphicx,hyperref}
\usepackage[capitalize,nameinlink]{cleveref}

\newtheorem{theorem}{Theorem}

\newtheorem{lemma}{Lemma}
\newtheorem{corollary}[lemma]{Corollary}

\title{Quasipolynomiality of the Smallest Missing Induced Subgraph}

\author{David Eppstein,\thanks{Computer Science Department, University of California, Irvine, eppstein@uci.edu. Research supported in part by NSF grant CCF-2212129.}
\quad Andrea Lincoln,\thanks{Boston University. Research performed while the author was at MIT and the University of California, Berkeley.}
\quad  and \quad Virginia Vassilevska Williams\thanks{EECS and CSAIL, Massachusetts Institute of Technology, virgi@mit.edu. Research supported in part by NSF Grant CCF-2129139.}}

\date{ }

\begin{document}

\maketitle
\begin{abstract}
We study the problem of finding the smallest graph that does not occur as an induced subgraph of a given graph. This missing induced subgraph has at most logarithmic size and can be found by a brute-force search, in an $n$-vertex graph, in time $n^{O(\log n)}$. We show that under the Exponential Time Hypothesis this quasipolynomial time bound is optimal. We also consider variations of the problem in which either the missing subgraph or the given graph comes from a restricted graph family; for instance, we prove that the smallest missing planar induced subgraph of a given planar graph can be found in polynomial time.
\end{abstract}

\section{Introduction}

In recent years, the older conventional wisdom that all natural computational problems have time either polynomial (or better) or exponential (or worse) has frayed, with the discovery of several natural problems whose best known time complexity is \emph{quasipolynomial}, of the form $n^{f(n)}$ for a polylogarithmically bounded function $f$. Of course, it is possible to obtain quasipolynomial time bounds from parameterized time bounds, by an appropriate choice of parameter, but instead we seek problems in which quasipolynomiality arises naturally from an unparameterized problem.
Such problems for which the best known time is quasipolynomial include:
\begin{itemize}
\item Finding minimum dominating sets in tournaments~\cite{MegVis-TCS-88}
\item Computing the Vapnik--Chervonenkis dimension of a family of sets~\cite{LinManRiv-IC-91,Shi-TCS-95,PapYan-JCSS-96}. Approximating VC dimension up to any constant factor \cite{Manurangsi23}.
\item Constructing $\epsilon$-approximate Nash equilibria, for constant $\varepsilon$~\cite{LipMarMeh-EC-03}.
\item Graph isomorphism and graph canonization~\cite{Bab-STOC-19}.
\item Maximum independent set in unit ball graphs in the hyperbolic plane~\cite{Kis-SODA-20}.
\end{itemize}
More strongly, in some of these cases, the form of the quasipolynomial time bound has been precisely determined (up to constant factors in the exponent) under the \emph{Exponential Time Hypothesis} (ETH) of Impagliazzo and Paturi~\cite{ImpPat-CCC-99}, the assumption that $3$-SAT cannot be solved in time $2^{o(n)}$. For instance, for both $\epsilon$-approximate Nash equilibria and hyperbolic unit ball independent sets, the upper bound is $n^{O(\log n)}$, while ETH implies the impossibility of improving this bound to $n^{o(\log n)}$~\cite{BraKunWei-SODA-15,Kis-SODA-20}.

In this work, we find another natural example of a problem that has both a quasipolynomial upper bound and a matching  lower bound derived from the Exponential Time Hypothesis. Our problem is the \emph{Smallest Missing Induced Subgraph}: given an undirected graph $G$ as input, find another undirected graph $H$, with as few vertices as possible, such that $H$ is not isomorphic to an induced subgraph of $G$. As we show, this problem can be solved in time $n^{O(\log n)}$ in $n$-vertex graphs~$G$ (\cref{thm:find-smis}), and improving this to time $n^{o(\log n)}$ would contradict ETH (\cref{thm:eth-hard}).

We also consider variations of the problem in which either the missing subgraph or the given graph comes from a restricted graph family. For instance, the smallest missing planar induced subgraph, in a given graph $G$, measures the quality of $G$ as a \emph{universal graph} for planar graphs and as the basis for adjacency-labeling schemes for planar graphs~\cite{DujEspGav-JACM-21}. We prove that the smallest missing planar induced subgraph of a given planar graph can be found in polynomial time (\cref{thm:planar-missing}).

\section{Upper bounds}

The Smallest Missing Induced Subgraph can be found by exactly counting the number of induced subgraphs of each type, of increasing size, until finding a size for which one of the counts is zero. Related problems of counting all types of subgraphs, exactly or approximately, have long been studied~\cite{DukLefRod-SICOMP-95,KloKraMul-IPL-00,RibParSil-CS-22}, and have recently gained popularity in bioinformatics, where small induced subgraphs are called ``graphlets''~\cite{Prz-BI-07,HocDem-BI-14,AhmNevRos-ICDM-15}. A naïve counting algorithm can solve this problem simply by enumerating all ordered $k$-tuples of vertices, determining the type of subgraph each $k$-tuple induces, and incrementing a counter for that type of subgraph.

In more detail:
\begin{itemize}
\item We begin with $k=2$, the smallest possible size of a missing induced subgraph in a graph with two or more vertices.
\item A $k$-vertex subgraph can be represented by a binary value with $\tbinom{k}{2}$ bits, representing (the lower half of) its adjacency matrix. We use these binary values as indices into an array of counters, initially all zero. We consider the vertices of the subgraphs to be labeled, so that permutations of the vertices are considered as different subgraphs.
\item We loop through all ordered $k$-tuples of vertices of the given graph. For each $k$-tuple, we can construct the binary value representing its adjacency matrix in average $O(k)$ time per $k$-tuple, by updating it from the previous $k$-tuple in the enumeration. After this step, we increment the counter indexed by that binary value. Note that by going through every $k$-tuple in the graph we are incrementing at least once for every possible re-labeling of each subgraph on $k$ nodes. In other words, we can not `miss' a valid labeling as we are going through all $n^k$ possible combinations of nodes. 
\item We loop through all cells of the array of counters, searching for one that is zero. If found, we decode it and output it as a missing subgraph; otherwise, we increment $k$ and repeat the previous steps.
\end{itemize}

The time for each $k$ is $O(kn^k+2^{k(k-1)/2})$. As this increases superexponentially with $k$, it is dominated by the last step, the one in which we find a missing induced subgraph.

\begin{lemma}
\label{lem:cutoff}
This algorithm succeeds for some $k\le 2\log_2n + 2$.
\end{lemma}

\begin{proof}
If it ever reaches $k > 2\log_2n + 1$, the number of array cell counters will be
\[ 2^{\binom{k}{2}} = 2^{k(k-1)/2} > 2^{k\log_2 n} = n^k. \]
There are more counters than increment steps, and some counter will remain unincremented. Therefore, the algorithm must necessarily terminate whenever $k$ becomes this large. Since $k$ is increased by one in each iteration, the value at which it exceeds this threshold is at most $2\log_2n + 2$.
\end{proof}

\begin{corollary}
\label{cor:smis-size}
In any $n$-vertex graph, the smallest missing induced subgraph has at most $2\log_2n + 2$ vertices.
\end{corollary}

Constructions of Alon~\cite{Alo-GFA-17} and Alstrup, Kaplan, Thorup, and Zwick~\cite{AlsKapTho-SIDMA-19} produce graphs in which the smallest missing induced subgraph has $2\log_2n - O(1)$ vertices, showing that this bound is nearly tight. By plugging this value of $k$ into the time bound for the algorithm as a function of $n$ and $k$, and simplifying the resulting expression, we obtain:

\begin{theorem}
\label{thm:find-smis}
The algorithm above takes time $O(n^{2\log_2n+3})$ to find the smallest missing induced subgraph of any $n$-vertex graph.
\end{theorem}

The bottleneck in this time bound is initializing and searching the table of counters. When $k=2\log_2n+2$, the size of this table is sufficiently larger than $n^k$ to make this step slower than the search through all $k$-tuples of vertices. The time can be reduced to $O(n^{2\log_2n+2}\log n)$, at the expense of randomization, by replacing the table of counters by a hash table; we omit the details.

The copies of any single $k$-vertex subgraph type can be counted in the substantially smaller time bound $O(n^{0.174k+o(k)})$~\cite{CurDelMar-STOC-17} but multiplying this bound by the number of distinct subgraph types would lead to a slower algorithm than the one above. Although there has been research on applying fast matrix multiplication to the induced subgraph counting problem for specific small values of $k$~\cite{KloKraMul-IPL-00,KowLinLun-SIDMA-13}, using sparsity parameters of the host graph to reduce the time for subgraph counting~\cite{EppSpi-JGAA-12,EppGooStr-TCS-12},
or using the structural features of specific types of subgraphs to speed up algorithms for counting only those subgraphs~\cite{WilWil-SICOMP-13}, we do not know of significant general speedups to the naïve counting algorithm above based on these ideas. Indeed, to make any significant improvement based on counting methods we would need not only algorithmic ideas, but a better bound than \cref{cor:smis-size} on the smallest $k$ for which success is guaranteed or a counting algorithm which takes less than $O(1)$ time per possible subgraph of size $k$. This is because, for the value of $k$ given by \cref{cor:smis-size}, our time bound is not significantly larger than what would be required merely to output all subgraph counts.  

\section{Lower bounds}

Our lower bound on the smallest missing induced subgraph follows from the hardness of finding cliques of logarithmic size, a special case of known results on the parameterized complexity of the clique problem~\cite{CheChoFel-IC-05,CheHuaKan-JCSS-06,LokMarSau-BEATCS-11}. We spell out the proof to clarify its applicability to the range of parameter values that we need. We follow the proof outline from Lokshtanov et al~\cite{LokMarSau-BEATCS-11}.

\begin{lemma}
\label{lem:log-clique}
Under the Exponential Time Hypothesis, for any constant $c>0$, it is impossible to find a maximum clique, in an $n$-vertex graph for which that clique is guaranteed to have size at most $c\log n$, in time $n^{o(\log n)}$.
\end{lemma}

\begin{proof}
We reduce 3-coloring of $N$-vertex graphs to clique-finding on exponentially-larger graphs, as follows. Choose a parameter $\delta>0$, and partition the given $N$-vertex graph $G$ into $t=\lfloor\delta\sqrt N\rfloor$ subsets of at most $\lceil N/t\rceil$ vertices. Form a graph $H$ representing  the possible 3-colorings of each subset, and add an edge in $H$ between pairs of compatible subset 3-colorings. (Two subset 3-colorings are compatible if they represent proper colorings of the induced subgraphs of different subsets, and they do not combine to form any monochromatic edge.) Then cliques of size $t$ in $H$ correspond to proper 3-colorings of $G$. The number $n$ of vertices in $H$ is $O(t3^{N/t})$, so $t=\Theta(\log n)$, with a constant of proportionality that can be made arbitrarily small by choosing a sufficiently small value of $\delta$, meeting the guarantee of the lemma. If we could find a $t$-clique in time $n^{o(\log n)}$, this would give us a 3-coloring algorithm with time $(t3^{N/t})^{o(t)}=3^{o(N)}$, contradicting the exponential-time hypothesis.
\end{proof}

To reduce from clique-finding to missing subgraphs, we use a construction for a family of graphs in which the unique smallest missing subgraph is a clique:

\begin{figure}[t]
\centering\includegraphics[width=0.6\textwidth]{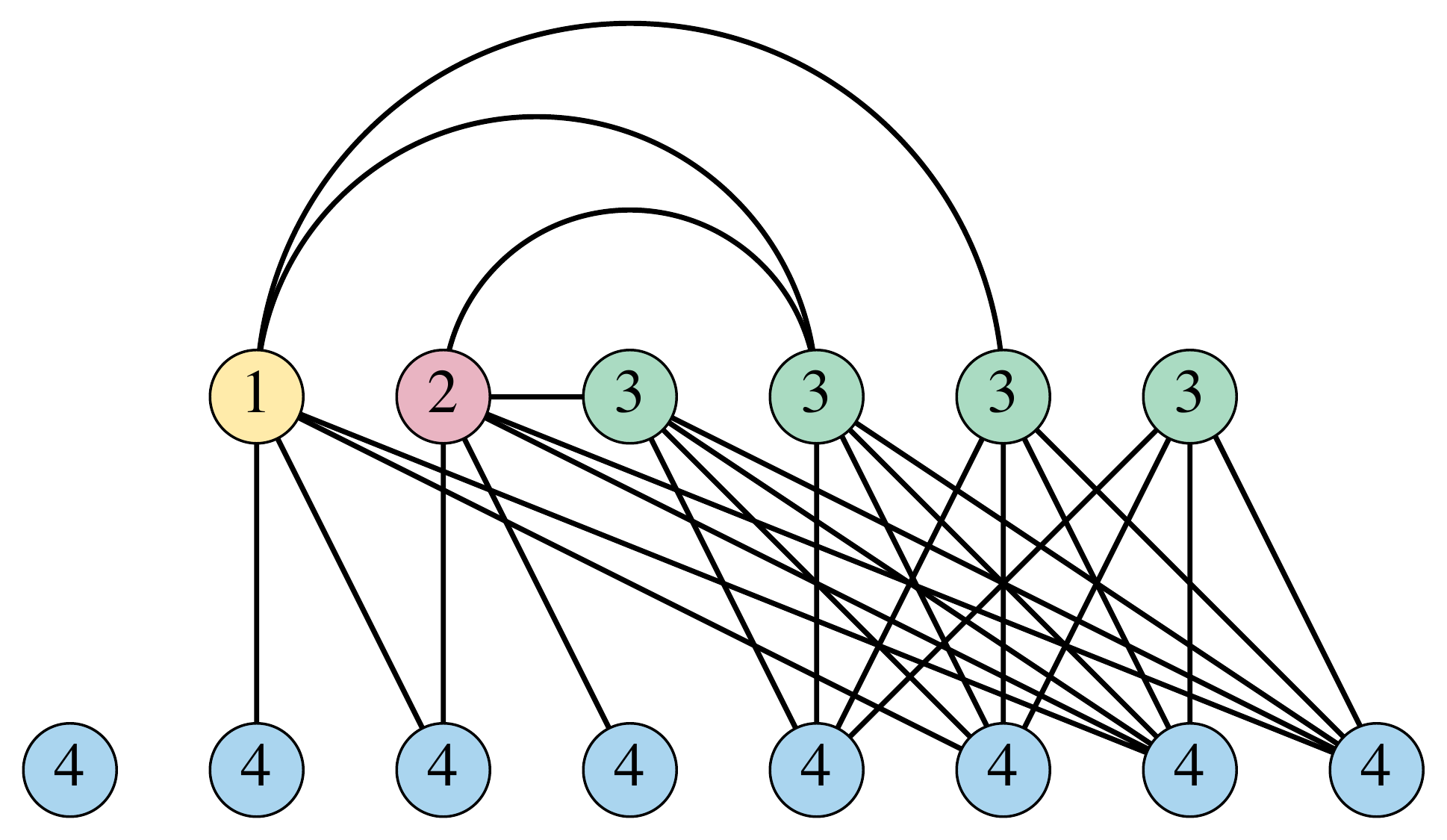}
\caption{The labeled graph $X_4$ of \cref{lem:all-but-clique}. Every graph with four or fewer vertices exists as an induced subgraph of $X_4$, except for a clique of size four. More strongly, whenever $H$ is a four-vertex graph $H$ labeled by the four numbers 1, 2, 3, 4 that does not have an edge labeled 1--2, $H$ appears as a labeled induced subgraph of $X_4$.}
\label{fig:all-but-clique}
\end{figure}

\begin{lemma}
\label{lem:all-but-clique}
For each positive integer $i>1$, there exists a graph $X_i$ that does not contain a clique of size $i$, but contains all other $i$-vertex induced subgraphs, and has a total of $2^i-2$ vertices.
\end{lemma}

\begin{proof}
More strongly, we form a family of labeled graphs $X_i$ with the following properties:
\begin{itemize}
\item $X_i$ has its vertices labeled by the numbers from $1$ to $i$.  It has one vertex labeled 1 and one vertex labled 2, and $2^{j-1}$ vertices labeled $j$ for $j>2$.
\item The vertices having the same label form an independent set, so each clique must have vertices of distinct labels.
\item There is no edge labeled 1--2, so it is not possible for both of these labels to appear in a clique. Therefore, the $i$-vertex clique is a missing induced subgraph.
\item If $H$ is any $i$-vertex graph $H$ labeled by distinct integers from $1$ to $i$, such that $H$ does not have an edge labeled 1--2, then $H$ appears with the same labels as an induced subgraph of $X_i$.
\end{itemize}
Every $i$-vertex graph that is not a clique can be labeled  by distinct integers from $1$ to $i$ so that its vertices labeled 1 and 2 are non-adjacent; therefore, for a graph $X_i$ with the properties listed above, every $i$-vertex graph that is not a clique appears as an induced subgraph, and the $i$-vertex clique is the unique smallest missing induced subgraph.

We construct the graphs $X_i$ inductively, as illustrated in \cref{fig:all-but-clique}. As a base case, for $i=2$, let $X_2$ consist of two isolated vertices, one with labeled $1$ and the other labeled $2$. For ease of notation call the set of all vertices labeled $i$ $T_i$. To form $X_i$ from $X_{i-1}$ add $2^{i-1}$ vertices, labeled $i$, one for each subset of $\{1,2,\dots,i-1\}$.
For a vertex $v$ labeled $i$ corresponding to the subset $S \in \{1,2,\dots,i-1\}$, add edges between $v$ and $u$ for all $u \in T_\ell$ for all $\ell \in S$. That is, connect $v$ to all verticies $u$ that have a label in $S$.

The resulting graph clearly has the first three properties listed above. To show that it has the last property, let $H$ be any $i$-vertex graph $H$ labeled by distinct integers from $1$ to $i$, such that $H$ does not have an edge labeled 1--2. Let $H^{-}$ be obtained from $H$ by removing the vertex labeled $i$, and let $S$ be the set of labels of neighbors of the removed vertex. Then by induction, $H^{-}$ appears as an induced subgraph of $X_{i-1}$, with the same labels. To form $H$ as an induced subgraph of $X_i$, add one more vertex to this copy of $H^{-}$, the vertex with label $i$ that corresponds to set $S$.
\end{proof}

Noga Alon (personal communication) has observed that a randomized construction from his work on induced-universal graphs~\cite{Alo-GFA-17} can be modified to produce smaller graphs with the same properties as in \cref{lem:all-but-clique}, with $2^{(i-1)/2}\bigl(1+o(1)\bigr)$ vertices. This size bound is within a lower-order additive term of optimal. However, for our purposes it is more important to have a deterministic construction than to optimize the exponent.

\begin{theorem}
\label{thm:eth-hard}
Under the exponential-time hypothesis, it is impossible to find the smallest missing induced subgraph of every $n$-vertex graph in time $n^{o(\log n)}$.
\end{theorem}

\begin{proof}
Let $X_i$ denote the family of graphs constructed by \cref{lem:all-but-clique}.
We can use an algorithm for the smallest missing induced subgraph to find a maximum clique of size up to $\log_2 n$, in a graph $G$ for which the maximum clique has that size, by successively finding the smallest missing induced subgraph in the disjoint union of graphs $G\cup X_i$ for $i=2,3,\dots$ until finding an $i$ for which this missing subgraph is an $i$-vertex clique. The maximum clique in $G$ has size one less than this stopping value of $i$. The addition of $X_i$ to $G$ at most doubles the number of vertices, which does not change the form of a time bound of form $n^{o(\log n)}$, and the loop through successive values of $i$ also does not change this form. Therefore, an algorithm for the smallest missing induced subgraph with time $n^{o(\log n)}$ would contradict \cref{lem:log-clique}.
\end{proof}

\section{Restricted missing subgraphs}

A natural extension of the smallest missing induced subgraph would be to search for the smallest missing induced subgraph from a given graph family $\mathcal{F}$, either in a larger graph from the same family or more generally. When $\mathcal{F}$ is totally ordered by the induced subgraph relation, this is equivalent to finding the largest induced subgraph in $\mathcal{F}$ in a given graph $G$; for instance, the classical problems of finding the largest clique, largest independent set, and longest induced path in a graph can all be expressed as a missing induced subgraph problem in this way.

For general graph families $\mathcal{F}$, the size of smallest missing induced subgraph in $\mathcal{F}$, in a graph~$G$, is one plus the largest $k$ for which $G$ is \emph{$k$-induced-universal}, meaning that it contains all $k$-vertex graphs in $\mathcal{F}$ as induced subgraphs. A universal graph in this sense is not required to belong to $\mathcal{F}$ itself. Universality is equivalent to the existence of an \emph{adjacency labeling scheme}, an assignment of binary labels to the vertices of graphs in $\mathcal{F}$ such that the adjacency of two vertices can be determined from their labels. If a $k$-universal graph exists, and has $n_k$ vertices, the identities of its vertices can be used as labels in an adjacency labeling scheme with $\lceil\log_2 n\rceil$ bits per label. Conversely, if an adjacency labeling scheme with $b$ bits per label is possible, then one can construct a $k$-universal graph with $2^b$ vertices, one for each possible label, with adjacency determined according to the labeling scheme. The computational problem of finding the smallest missing induced subgraph from $\mathcal{F}$ in $G$, then, is the same as determining the $k$ for which $G$ is $k$-universal. If the smallest missing subgraph from $\mathcal{F}$ in $G$ is of size $k+1$ then all subgraphs of size $k$ appear $G$ making $G$ $k$-universal (for the family $\mathcal{F}$). Thus, known upper bounds on the size of universal graphs translate into lower bounds on the possible size of the smallest missing subgraph. For instance, trees have universal graphs of linear size, from which it follows that the smallest missing tree, in a given graph, may also have linear size~\cite{Chu-JGT-90}. Translating other known results on universal graphs into this form, it follows that the smallest missing induced subgraph of maximum degree $\Delta$ can have size $\Theta(n^{2/\Delta})$~\cite{AloNen-MPCPS-19} and that the smallest missing planar induced subgraph can have size $n^{1-o(1)}$~\cite{DujEspGav-JACM-21}.

Additional problems of this type arise when both the induced subgraph and $G$ are restricted to belong to the same family $\mathcal{F}$. For instance, for the smallest missing induced bipartite subgraph of a given bipartite graph, a counting argument similar to the one we used for general graphs shows that this subgraph has size $O(\log n)$, from which it follows that we can find this subgraph by a brute-force search in quasipolynomial time. However, our ability to prove the optimality of this time bound is limited in two ways: we do not have a construction for a bipartite graph in which the unique smallest missing induced bipartite subgraph is a biclique, analogous to our construction for cliques, and we do not know ETH-tight bounds on the complexity of finding bicliques in bipartite graphs~\cite{Lin-JACM-18}.

Another interesting case is the problem of finding the smallest missing induced planar subgraph, in a given planar graph. Like the unconstrained smallest missing induced subgraph, the smallest planar missing induced subgraph in a planar graph has logarithmic size. However, in this case the search for the subgraph can take advantage of planarity, speeding up the overall algorithm to polynomial time. In the remainder of this section we detail this result.

\begin{lemma}
\label{lem:planar-log}
The smallest missing induced planar subgraph, in a given $n$-vertex planar graph, has $O(\log n)$ vertices.
\end{lemma}

\begin{proof}
The high level idea is to compare the exponential number of 4-connected planar triangulations (that is, maximal planar graphs) of size $k$, on the one hand, and an $O(n)$ bound on the number of 4-connected triangulations that can exist as induced subgraphs of a given $n$-vertex graph, on the other hand. When all $k$-vertex 4-connected planar triangulations are present in a graph, their exponential number must be $O(n)$, from which it follows that $k$ must be $O(\log n)$. Although the smallest missing induced subgraph might not be 4-connected or maximal planar, it has at most $k$ vertices.

In more detail, the number of 4-connected planar triangulations on $k$ unlabeled vertices can be upper-bounded by $c^k$ for some constant~$c$~\cite{Tut-CJM-62,Rat-JCTB-74}. However, a given planar graph $G$, with $n$ vertices, can only have $O(n)$ induced subgraphs that are 4-connected and maximal planar. This is because each such subgraph is one of the linearly many graphs formed when $G$ is partitioned along its separating triangles, which necessarily form a nested family of separations~\cite{EppRee-SODA-19}. This partition cannot subdivide any 4-connected triangulation, because if it could be subdivided in this way it would not be 4-connected. Conversely, if $G$ contains a 4-connected triangulation as an induced subgraph, each face of the triangulation must either be a face of $G$ or a separating triangle, so this triangulation will be separated from all the other vertices of $G$ by the partition of $G$ on its separating triangles. Putting the two bounds on the number of $k$-vertex triangulations and on the number of induced 4-connected triangulations together, if $G$ is an $n$-vertex graph for which all $k$-vertex planar graphs are present as induced subgraphs, it must be the case that $c^k=O(n)$. Taking logarithms of both sides shows that $k=O(\log n)$. For the induced subgraphs of any larger size, at least one is missing.
\end{proof}

Although we do not need it, we note that conversely there exist planar graphs whose smallest missing induced subgraph has size $\Omega(\log n)$. This follows from the fact that the number of $k$-vertex unlabeled planar graphs of all types is only singly exponential in $k$~\cite{GimNoy-JAMS-09}. An $n$-vertex planar graph whose smallest missing planar induced subgraph is logarithmic can be constructed by choosing the largest $k$ such that the disjoint union of all $k$-vertex planar graphs has at most $n$ vertices, and then padding this disjoint union to exactly $n$ vertices by adding isolated vertices.

\begin{theorem}
\label{thm:planar-missing}
The smallest missing induced planar subgraph, in a given $n$-vertex planar graph $G$, can be found in time polynomial in~$n$.
\end{theorem}

\begin{proof}
For each positive integer $k$, in numerical order, we list all $k$-vertex planar graphs and test for each one whether it is an induced subgraph of $G$. The $k$-vertex unlabeled planar graphs can be listed in total time exponential in $k$ using standard planar graph enumeration methods~\cite{BriMcK-MATCH-07,BosHur-CG-09}. An algorithm of Eppstein~\cite{Epp-JGAA-99}, subsequently improved by Dorn~\cite{Dor-STACS-10} and Bonsma~\cite{Bon-STACS-12}, allows testing whether a graph $H$ of size $k$  is a subgraph or induced subgraph of a graph $G$ of size $n$, in time $2^{O(k)}n$. Because we only apply these algorithms for $k=O(\log n)$, the total time bound is polynomial in~$n$.
\end{proof}

The constant factor in the logarithmic bound on the size of the smallest missing induced planar subgraph that can be obtained from our proof of \cref{lem:planar-log} is unlikely to be tight, and the base of the exponential function in the $2^{O(k)}n$ time bound for planar subgraph isomorphism is large (approximately $76^k$ in Bonsma's version). For this reason we have not analyzed the time bound of \cref{thm:planar-missing} more carefully to determine its exact exponent.

\section{Conclusions and Related Problems}

We have shown that, to find the smallest missing induced subgraph of a given $n$-vertex graph, it is possible to use an algorithm with running time $n^{O(\log n)}$, and that under standard assumptions no better time bound is possible. We have also investigated similar problems where the induced subgraph or the given graph must belong to a more restricted family of graphs.

More generally, the problem of asking for the smallest missing subobject of an object can be extended to many other inclusion orderings of other types of objects. For substrings or subsequences, it is not difficult to solve in polynomial time. A more interesting example of the same type of problem arises for permutations and permutation patterns, smaller permutations with the same ordering as a subsequence of a larger permutation. A permutation that contains all patterns of length up to some value $k$ is known as a \emph{superpattern}, and so the problem of finding the smallest missing permutation is equivalent to asking for what $k$ a given pattern is a superpattern. A counting argument shows that a $k$-superpattern must have length $\Omega(k^2)$, and superpatterns of this length are known~\cite{Arr-EJC-99,Mil-JCTA-09,ChrKwaSin-JCTA-21,EngVat-AMM-21}. Turning this around, any permutation of length $n$ must have a missing pattern of length $O(\sqrt n)$, from which it follows that this missing pattern can be found in time $n^{O(\sqrt n)}$. Is a time bound of this form tight?

\bibliographystyle{plainurl}
\bibliography{qpoly}
\end{document}